\RequirePackage{fix-cm}
\documentclass[smallextended]{svjour3}       
\smartqed  
\usepackage{graphicx}
\begin{document}

\title{Minimal Logarithmic Signatures for one type of Classical Groups
}
\subtitle{MLSs for one type of Classical Groups}


\author{  Haibo Hong  \and  Licheng Wang  \and Haseeb Ahmad \and Yixian Yang
}


\institute{Information Security Center, State Key Laboratory of
Networking and Switching Technology, Beijing University of Posts and
Telecommunications, Beijing, 100876 P.R. China  \email{  honghhaibo1985@163.com \quad wanglc@bupt.edu.cn}           
}

\date{Received: date / Accepted: date}

\maketitle

\begin{abstract}
As a special type of factorization of finite groups, logarithmic signature (LS) is used as the main component of cryptographic keys for secret key cryptosystems such as PGM and public key cryptosystems like $MST_1$, $MST_2$ and $MST_3$. An LS with the shortest length, called a minimal logarithmic signature (MLS), is even desirable for cryptographic applications. The MLS conjecture states that every finite simple group has
an MLS. Recently, the conjecture has been shown to be true for general linear groups $GL_n(q)$,
special linear groups $SL_n(q)$, and symplectic groups $Sp_n(q)$
with $q$ a power of primes and for orthogonal groups $O_n(q)$ with $q$ as a power of 2. In
this paper, we present new constructions of minimal logarithmic
signatures for the orthogonal group $O_n(q)$ and $SO_n(q)$ with $q$ as a power of odd primes.
Furthermore, we give constructions of MLSs for a type of classical groups --- projective commutator subgroup $P\Omega_n(q)$.


\end{abstract}

\section{Introduction}
The security of many public key cryptosystems is based on the hardness assumptions of certain problems over large finite abelian algebraic structures such as cyclic groups, rings and finite fields. Two well-known hard problems are integers factorization problem (IFP) and discrete logarithm problem (DLP). However, these hardness assumptions would be broken if quantum computers become practical. For instance, Shor's quantum algorithms \cite{S97} solve IFP and DLP very efficiently. The security status of currently used cryptosystems, mainly based on IFP and DLP or their variants, becomes even worse due to the known great progress on finding possible solutions for building quantum computers on practical scales.

Therefore, it is imminent to design effective and practical cryptographic schemes that have the potential for resisting quantum algorithm attacks. Actually, several attempts using non-abelian algebraic structures were made and some available cryptographic schemes such as $PGM$, $MST_1$, $MST_2$ and $MST_3$
\cite{BSGM05,MST02,MM92,M02,LMTW09,MST12} were developed during the past decades. In particular, as a natural analogy of the hardness assumption of IFP, the group factorization problem (GFP)\cite{M86,QV94} and its hardness assumption over certain factorization basis, referred as logarithmic signature, play a core role in the security arguments for the family of $MST$ cryptosystems.

Security is not the unique goal of designing a cryptosytem. Instead, efficiency is also a major issue. With the purpose for minimizing the parameter sizes, a natural question comes to mind: How to make the factorization basis known as logarithmic signature (LS), as short as possible in the family of MST cryptosystems? A minimal logarithmic signature (MLS) is an LS with the shortest length. In other word, further shorten an MLS would make it no longer an LS. New question arises: Does any finite (non-abelian) group has MLSs?

In fact, some encouraging work has been done in searching the MLSs for finite groups.
According to the pioneering work due to Vasco et al. \cite{GRS03}, Holmes \cite{H04} and Lempken et al. \cite{LT05}, we know that, with few exceptions, MLSs exist for all groups of order $\leq 10^{10}$. Most recently, Nikhil Singhi, Nidhi Singhi, and Magliveras \cite{NNM10,NN11} made another breakthrough: MLSs exist for the groups $GL_n(q)$, $SL_n(q)$, $Sp_n(q)$ with $q$ as a power of a prime
and $O_{n}(q)$ with $q$ as a power of 2. As far as we know, this is the first result not constrained by a specified boundary on group orders. Besides, Nikhil Singhi and Nidhi Singhi \cite{NN11} also pointed out, without any proof, that the MLSs \emph{should also exist} for $O_{n}(q)$ with $q$ as a power of odd primes.

Therefore, in this paper, our main motivation is to present new constructions of minimal logarithmic signatures for the orthogonal group $O_n(q)$, the special orthogonal group $SO_n(q)$, the projective special orthogonal group $PSO_n(q)$, the commutator subgroup  $\Omega_n(q)$ and one type of classical groups \cite{G97} ---  projective commutator subgroup $P\Omega_n(q)$ with $q$ as a power of odd primes. For $O_n(q)$ and $SO_n(q)$, the proposed MLSs have the similar structure $[A, B', G_w]$, where $A=\langle a\rangle$, and $B'=\{hC\mid h\in B\}$, $C \leq B$, $B=\langle b\rangle$, while $G_w=P:Q$ is a semi-direct product of a $p$-group $P$ and a direct product $Q=GL_1(q)\times Y$ (see Table 1).  We employ two canonical homomorphisms $\eta: SO_n(q)\to PSO_n(q)$ and $\theta: \Omega_n(q)\to P\Omega_n(q)$ for proving the existence of MLSs for $PSO_n(q)$ and $P\Omega_n(q)$, respectively.

\begin{table}[htbp]
\begin{center}
\footnotesize
\caption{MLSs for $O_n(q)$ and $SO_n(q)$  }\label{tbl:MLS}
\begin{tabular}{|c|c|c|c|c|}
\hline
                           &  & $B'=\{hC\mid h\in B\}$ with &  \multicolumn{2}{c|}{$G_w=P:Q$ with}  \\
                          & $A=\langle a\rangle$  & $C\leq B=\langle b\rangle$,$|C|=q-1$   &  \multicolumn{2}{c|}{$Q=GL_1(q)\times Y $}  \\
    \cline {2-5}
                           & $a$   &   $b$  &   $|P|$  &     $Y$  \\
    \cline {1-5}

          & $ x_1^{q^m-1}$ for $x_1\in GL_{2m}(q)$ &   $\left(
  \begin{array}{cccc}
   D_1 & 0 & 0 & 0 \\
      0  & 1 & 0 & 0 \\
      0  & 0 & (D_1^t)^{-1} & 0 \\
      0  & 0 & 0 & 1\\
  \end{array}
\right) $  for  $D_1 \in GL_{2m-2}(q)$  & $q^{2m-2}$  &  $O_{2m-2}^{-}(q)$ \\

\cline{2-5}

 $O_n(q)$       & $x_2^{q^{m-1}-1}$ for  $x_2 \in GL_{2m}(q)$ &   $\left(
                                                                   \begin{array}{cc}
                                                                   D_2  & 0  \\
                                                                    0     &   (D_2^t)^{-1}\\
                                                                    \end{array}
                                                                  \right)$ for  $ D_2 \in GL_{2m}(q)$ & $q^{2m-2}$  &  $O_{2m-2}^{+}(q)$ \\

 \cline{2-5}

               &  $x_3^{q^m-1}$ for  $x_3 \in GL_{2m+1}(q)$ &  $ \left(
                                                                  \begin{array}{ccc}
                                                                  D_3   & 0 & 0 \\
                                                                    0 & (D_3^t)^{-1} & 0 \\
                                                                    0 & 0  & 1 \\
                                                                  \end{array}
                                                                \right)$  for  $D_3\in GL_{2m}(q)$ & $q^{2m-1}$ &  $O_{2m-1}(q)$ \\
 \cline{1-5}

 & $ x_1^{*q^m-1}$ for $x_1^*\in O^{-}_{2m}(q)$ &

 $\left( \begin{array}{cccc}
      D_1^* & 0 & 0 & 0 \\
     0 & 1 & 0 & 0 \\
     0 & 0 & D_1^{*t} & 0 \\
     0 & 0 & 0 & 1 \\
   \end{array}
 \right)$  for  $D_1^* \in O_{m-1}(q)$  & $q^{2m-2}$  &  $SO_{2m-2}^{-}(q)$ \\

\cline{2-5}

 $SO_n(q)$       & $x_2^{*q^{m-1}-1}$ for  $x_2^* \in O^{+}_{2m}(q)$ &

$ \left(  \begin{array}{cc}
    D_2^*  &  0\\
          0    &   D_2^{*t} \\
   \end{array}
 \right)$    for  $ D_2^* \in O_{m}(q)$ & $q^{2m-2}$  &  $SO_{2m-2}^{+}(q)$ \\

 \cline{2-5}

               &  $x_3^{*q^m-1}$ for  $x_3^* \in O_{2m+1}(q)$ &

          $\left( \begin{array}{ccc}
              D_3^*  & 0 & 0 \\
              0 & D_3^{*t}  & 0 \\
              0 & 0 & 1 \\
            \end{array}
          \right)$    for  $D_3^* \in O_{m}(q)$ & $q^{2m-1}$ &  $SO_{2m-1}(q)$ \\

\hline
\end{tabular}
\end{center}
\end{table}
\normalsize

The rest of contents are organized as follows: Necessary preliminaries are presented in Section 2; In Section 3, we utilize the
Levi decomposition of parabolic subgroups to construct LSs for the parabolic subgroups of $O_n(q)$ and $SO_n(q)$; In Section 4, analogous to methods in \cite{NN11}, we use the totally isotropic subspaces to prove the existence of MLSs for $O_{2m}^{-}(q)$ and $SO_{2m}^{-}(q)$; In Section 5, we utilize  suitable spread in the $P(V)$ \cite{BKMS08,K82,T81} to accomplish the proof for $O_{2m}^{+}(q)$ and $SO_{2m}^{+}(q)$; In Section 6, we make further efforts to present the constructions of MLSs for $PSO^{\pm}_{2m}(q)$ and $P\Omega^{\pm}_{2m}(q)$; In Section 7, we take account of MLSs for $O_{2m+1}(q)$, $SO_{2m+1}(q)$, $PSO_{2m+1}(q)$, $\Omega_{2m+1}(q)$ and $P\Omega_{2m+1}(q)$.

\section{ Preliminaries }

\subsection{Classical Spreads and Quadratic Spaces in Finite Fields}

Let $K$ be a finite field and $V$ a $n$-dimensional vector space over $K$. For $v_1 \in V$, $\langle
v_1\rangle$ denotes the one-dimensional subspace generated by $v_1$. $P(V)$
denotes the projective space on $V$, which is the set of all one-dimensional subspaces of $V$ \cite{W09}.

Now,  we describe the classical spread \cite{T81,K82,D86,BKMS08,NN11}.
 An \emph{r-partial spread} in $V$ is a set $S = \{W_i\mid 1 \leq i \leq t\}$ of
$r$-dimensional subspaces $W_i$ such that $W_i \cap
W_j = \langle 0 \rangle$ for $i \neq j$. If $\cup_{i=1}^{t} W_i = V$, then $S$ is an \emph{r-spread} in $V$. Besides, when $S$ is an $r$-(partial) spread in $P(V)$, it partitions $P(V )$ into $(r-1)$-dimensional subspaces of $P(V)$.

 Suppose that $V = F_{q^{2m}}$ is a finite field, $\alpha$ is a primitive element of field $F_{q^{2m}}$ and $W = F_{q^{m}}$ is an $m$-subspace of $V$. For every $x \in V$, $W_x = \{wx \mid w \in W\}$,
%
 the set $S=\{W_x \mid x\in V\}$ forms a $m$-spread in $V$ \cite{NN11}. Meanwhile, we have the following remark.

\begin{remark} \cite{NNM10,NN11} \label{remark:Wi}
 Let $W_i=W\alpha^{(q^m-1)i} = \{w\alpha^{(q^m-1)i} | w \in W\}$, $0
\leq i \leq q^m$. Then, the spread $S$ can also be described as
$S = \{W_i | 0 \leq i \leq q^m\}$.
\end{remark}

Let $V$ be an $n$-dimensional vector space over the finite field $K = F_q$ with $q=p^e$ for some prime $p$ and a positive integer $e$. Let $\mathcal {B} = \{e_1,\cdots, e_n\}$ be an ordered basis for $V$. Then a\emph{ bilinear form} over a vector space $V$
is a map $f : V \times V \rightarrow K$
satisfying:
\begin{eqnarray*}
  f(\lambda u + v,w)&=& \lambda f(u,w) + f(v,w),  \\
  f(u,\lambda v +w) &=& \lambda f(u,v)+f(u,w).
\end{eqnarray*}
The radical of $f$, denoted by \emph{rad(f)}, is $V^{\bot}= \{u \in V \mid f(u, v) = 0, \forall v \in V\}$. $f$ is called \emph{non-singular} if $rad( f ) = \langle 0\rangle$.

A map $Q: V\rightarrow K $ is called a \emph{quadratic form}
if it satisfies:
\begin{eqnarray*}
  Q(\lambda u + v) = \lambda^2 Q(u) + \lambda f(u, v) + Q(v),
\end{eqnarray*}
where $f$ is a symmetric bilinear form.The radical of $Q$ is $rad(Q)= \{v \in rad( f ) \mid Q(v) = 0\}$. $Q$ is called \emph{non-singular} if $rad(Q) =\langle 0\rangle$.

 An \emph{isometry} on a quadratic space $(V, Q)$ is a non-singular linear map $g : V
\rightarrow V$ such that $Q(g(v)) = Q(v)$ for all $v \in V$. Two
 quadratic spaces $(V, Q)$ and $(V, Q')$ are said to be \emph{equivalent}, if there is an isometry $g : V \rightarrow V$ \cite{W09,NN11}. Besides, the group of all isometries of an inner-product space $(V,f)$ is denoted by \emph{Isom(V, f)} and that of all isometries of a quadratic space $(V, Q)$ is denoted by \emph{Isom(V, Q)}. When $q$ is a power of odd primes then we have, \emph{Isom(V, f)}= \emph{Isom(V, Q)} \cite{W09,NN11}.

A vector $v \in V$ is said to be \emph{isotropic} if $f(v, v) = 0$,
\emph{singular} if $Q(v) = 0$ and \emph{non-singular} if $Q(v)\neq 0$. A
subspace $W$ of $V$ is called \emph{totally isotropic} if $f(u, v) = 0$ for
all $u$, $v \in W$ and \emph{totally singular} if $Q(v)= 0$ for all $v \in W$.
Besides, a point $\langle v\rangle \in P(V )$ is called a \emph{singular point}
if $v$ is a singular vector, and $\langle v\rangle$ is called an \emph{isotropic point} in $P(V )$ if $v$ is an isotropic vector.

\subsection{Logarithmic Signatures and  Minimal Logarithmic Signatures for Finite Groups}


\begin{definition}[Logarithmic Signature]\cite{LT05}
Let $G$ be a finite group, $A\subseteq G$. Let $\alpha =[A_1, \cdots, A_s]$ be a sequence of ordered subsets $A_i$ of $G$ such that $A_i = [a_{i1}, \cdots, a_{ir_i}]$ with $a_{ij} \in G$~$(1 \leq j \leq r_i)$. Then $\alpha$ is called a logarithmic signature for $G$ (or $A$) if each
$g \in G$ (or $A$) is uniquely represented as a product
 \begin{center}
  $g = a_{1j_1}\cdots a_{sj_s}$
\end{center}
 with $a_{ij_i}\in A_i~(1 \leq i \leq s)$.
\end{definition}

The sequences $A_i$ are called the blocks of $\alpha$, the length of $\alpha$ is defined to be $l(\alpha) =\sum_{i=1}^s
 r_i$. Let $|G| =\prod_{j=1}^{k} p_{j}^{a_j}$~(or $|A| =\prod_{j=1}^{k}p_{j}^{a_j}$) be the prime power
decomposition of $|G|$~(or $|A|$) and $\alpha = [A_1, A_2, \ldots, A_s]$ be an LS for $G$~(or $A$). From \cite{GS02}, we have $l(\alpha) \geq \sum_{j=1}^{k} a_{j} p_j$.

\begin{definition}[Minimal Logarithmic Signature]\cite{LT05}
A logarithmic signature $\alpha$ for a finite group $G$~(or $A$) with $l(\alpha) = \sum_{j=1}^{k} a_{j} p_j$
 is called a minimal logarithmic signature (MLS) for $G$~(or $A$).
\end{definition}

\begin{lemma}\cite{NNM10} \label{le:1}
Let $A, B \leq G$, if $A$ and $B$ satisfy any one of the following two conditions:
\begin{enumerate}
\item[(i)]  $G = A\times B$ is a direct product of $A$ and $B$, $A \cap B = \{1\}$
\item[(ii)] $G = A:B$ is a semi-direct product of $A$ and $B$, $A \cap B = \{1\}$.
\end{enumerate}
Then, [A, B] is an LS for $G$.

\end{lemma}

\begin{lemma} \cite{NNM10,NN11}\label{le:2}
Let $H$ be a normal subgroup of $G$, $A \subseteq G$ and $\eta$ the canonical homomorphism $\eta:G\rightarrow G/H$ such that $a$, $b \in A$, $a \neq b$ imply that $aH \neq bH$. Let $A' = \eta(A)$, and suppose hat $[A_1, A_2,\cdots, A_k]$ is an LS for $A$. Let $B_i = \eta(A_i)
\subseteq G/H$ for $1 \leq i \leq k$. Then, $[B_1, B_2, \cdots, B_k]$
is an LS for $A'$.
\end{lemma}

Now, let $V$ be a finite dimensional vector space over $F_q$, $f$ be a bilinear form and $Q$ be a quadratic form. We call $L\subseteq P(V)$ a \emph{Singhi subset} \cite{NN11} if $L$ is one of the following sets \cite{NNM10,NN11}:
\begin{enumerate}
\item[(i)] the set of all isotropic points of $P(V)$ with respect to the bilinear form $f$,
\item[(ii)] the set of all non-isotropic points of $P(V)$ with respect to the bilinear form $f$,
\item[(iii)] the set of all singular points of $P(V)$ with respect to the quadratic form $Q$,
\item[(iv)] the set of all non-singular points of $P(V)$ with respect to the quadratic form $Q$.
\end{enumerate}
(Note that a \emph{Singhi subset} $L$ that meets condition (i) will be used in our later construction.)

\begin{lemma}\cite{NNM10,NN11} \label{le:3}
 Suppose that $G|L$ is a transitive permutation group action such that $G$
is a subgroup of $GL(V)$ and $L \subseteq P(V )$ is a Singhi subset.
 Let $S$ be an $r$-partial spread in $V$, which partitions $L$. Let $W \in S$, $w \in P(W)$ and $G_w$ be
the stabilizer of $w$ in $G$. Suppose there are sets $A$, $B
\subseteq G$ such that
\begin{enumerate}
\item[(i)] $A$ acts sharply transitive on $S$ with respect to $W$ under the action of $G$ on the set of all $r$-dimensional subspaces of
$V$.
\item[(ii)] $B$ acts sharply transitive on $L \cap P(W)$ with respect to
$w$ under the action of $G$ on $P(W)$.
\end{enumerate}
Then, $[A, B, G_w]$ is an LS for $G$.
\end{lemma}

\begin{lemma}\cite{GS02,NNM10,NN11}\label{le:4}
If $G$ is a solvable group, then $G$ has an MLS.
\end{lemma}

\begin{lemma}\cite{NN11}\label{le:5}
Let $G$ be a finite group and $x \in G$ be an element of order $t$. For $s \in N, s \leq t$, let $S = \{x^i | 0 \leq i < s\}=\{1,x^1,x^2,\cdots,x^{s-1}\}$ be a cyclic set generated by $x$. Then $S$ has an MLS $\beta=[A_1,A_2,\cdots, A_k]$ satisfying the following condition:

\begin{center}
      For any list $[j_i,j_2,\cdots, j_k]$, such that $x^{j_i}\in A_i, 1\leq i\leq k$, $\sum_{i=1}^{k} j_i <s$.
\end{center}

\end{lemma}

\begin{lemma}\cite{NN11}\label{le:6}
Let $G$ be a finite group and $[A_1, \cdots, A_r]$ be an LS for $G$ such that for each
subset $A_j$, $1 \leq j \leq r$, an MLS exists. Then $G$ has an MLS.
\end{lemma}

\section{Construction I: LSs for Parabolic Subgroups of $O_n(q)$ and $SO_n(q)$}
Throughout this paper, we assume that $q$ is a power of odd primes.  First, we construct the LSs for parabolic subgroups in
$O_{2m+1}(q)$ and $O_{2m}^{\pm}(q)$.

Let $W$ be an isotropic $k$-space of $V= F_{q^{2m}}$, then the stabilizer of $W$ is the maximal parabolic subgroup $P_{max}$ in $O_ {2m}^{\pm}(q)$ of shape $q^{k(k-1)/2+k(2m-2k)}:(GL_k(q)\times O_{2m-2k}^{\pm}(q))$.
%
Specifically, $P_{max}$ has the shape
\begin{center}
$\left(\begin{array}{ccccc}
       1&      0&       0&  \cdots& 0 \\
\vdots& \ddots& \vdots &        &\vdots\\
\ast  &\cdots &  1     & \cdots & 0\\
  0   &   0   &  0     & A      & 0\\
  0   &   0   &  0     & C      & D\\
  \end{array}
\right)$
\end{center}  the normal $p$-subgroup $R$ is a group of shape
\begin{center}
$\left(
  \begin{array}{ccccc}
     1&      0&       0&  \cdots& 0 \\
\vdots& \ddots& \vdots &        &\vdots\\
\ast  &\cdots &    1   & \cdots & 0\\
   0  &    0  &    0   & I_k    & 0\\
   0  &    0  &    0   & C'  & I_{2m-2k}\\
  \end{array}
\right)$
\end{center}
with center of order $q^{k(k-1)/2}$ and the subset $Q$ of matrices
of the shape

\begin{center}
$\left(
  \begin{array}{ccc}
     I_{k}  &     0  &  0 \\
 0      & A      &  0  \\
 0      & 0      &  D  \\
  \end{array}
\right)$

\end{center}
is a subgroup isomorphic to $GL_k(q)\times O_{2m-2k}^{\pm}(q)$.
Moreover, $R \cap Q =\{I_n\}$, $P_{max}=R:Q$, therefore, $P_{max}$ has an LS $[R,Q]$ (see Lemma \ref{le:1}).
Also, $P_{max}$ is of shape $q^{k(k-1)/2+k(2m+1-2k)}:(GL_k(q)\times O_{2m+1-2k}(q))$ in $O_{2m+1}(q)$.

Similarly, for $SO_{2m}^{\pm}(q)$ and $SO_{2m+1}(q)$, $P_{max}'$ is of shape $q^{k(k-1)/2+k(2m-2k)}:(GL_k(q)\times SO_{2m-2k}^{\pm}(q))$ in $SO_{2m}^{\pm}(q)$ and $q^{k(k-1)/2+k(2m+1-2k)}:(GL_k(q)\times SO_{2m+1-2k}(q))$ in $SO_{2m+1}(q)$.

\section{Construction II: MLSs for $O_{2m}^{-}(q)$ and $SO_{2m}^{-}(q)$}

Now, we construct MLSs for $O_{2m}^{-}(q)$ and $SO_{2m}^{-}(q)$. Here, our fundamental tools are Lemma \ref{le:3} and Lemma \ref{le:6}. Choosing suitable quadratic form $Q$ of minus type, we  take advantage of  all isotropic points of $P(V)$ for constructing the corresponding MLSs.

First, we observe $O_{2m}^{-}(q)$. Suppose $q$ is a power of odd primes, $V = F_{q^{2m}}$ is a $2m$-dimensional vector space over
$F_q$ , and $L$ is the set of all isotropic points of $P(V )$ . For $y
\in V$, $\overline{y}$ denotes $y^{q^m}$. $T_s : V \rightarrow V$ is a linear
transformation defined by $T_s(v) = sv$ for a given $s\in V$ and all $v \in V$. Let $\alpha$ be a primitive element of the field $F_{q^{2m}}$ and $x \in GL_{2m}(q)$ be the matrix corresponding to
the linear transformation $T_\alpha$ \cite{NN11}.  We define a
bilinear map $f : V \times V \rightarrow F_q $ by $f (x, y) =
tr_{F_{q^m} /F_q} (x \overline{y}+\overline{x} y) = \sum_{i=0}^{m-1}
(x \overline{y} + \overline{x} y)^{q^i}$
and a quadratic form $Q : V \rightarrow F_q $ by $Q(x) =
tr_{F_{q^m} /F_q} (x \overline{x}) = \sum_{i=0}^{m-1} (x\overline{x}
)^{q^i}$ \cite{NNM10,NN11}. Then, we observe that the number of
non-zero isotropic points in $P(V)$ with respect to the quadratic
space $(V, Q)$ are $(q^m + 1)(q^{m-1} - 1)/(q - 1)$ and the quadratic form $Q$ is of minus type \cite{W09,NN11}. Let $G$ be the
group of all isometries of $(V, Q)$, then $ G \cong O_{2m}^{-}(q)$
and $G$ is a permutation group acting transitively on isotropic
points \cite{W09,NN11}.

Now, we roughly explain the main idea for constructing the MLSs. As
described above, since the number of non-zero isotropic points in
$P(V)$ with respect to the quadratic space $(V, Q)$ are $(q^m +
1)(q^{m-1} - 1)/(q - 1)$, therefore, we need to construct a cyclic group $A$ of order
$q^m + 1$, which must be sharply transitive on a partial spread $S$ and a cyclic
set $B$ of cardinality $(q^{m-1} - 1)/(q - 1)$, which must be sharply
transitive on the projective subspace of $P(V)$. Then, we take
advantage of the Lemma \ref{le:3}, Lemma \ref{le:5} and Lemma \ref{le:6} for proving the existence of MLSs for $O_{2m}^{-}(q)$.

First, we define two special cyclic subgroups of
$O_{2m}^{-}(q)$. Let $a_1= x_1^{q^m-1} \in GL_{2m}(q)$ be the matrix
corresponding to the linear transformation $T_{\alpha^{q^m-1}}$.
Let $W_1'=\{e_1,e_2,\cdots,e_{m-1}\}$ is an $(m-1)$-dimensional totally isotropic subspace of $V$,  $D_1 \in GL(W_1')$ be a generator of the \emph{Singer cyclic subgroup} of $GL(W_1')$ \cite{CD04} .
Then $b_1 \in GL_{2m}(q)$ can be well defined as
follows:

\begin{center}

$b_1$=$\left(
  \begin{array}{cccc}
    D_1& 0 &0 &0           \\
 0& 1 &0  &0 \\
 0& 0 & (D_1^t)^{-1}&0  \\
 0 & 0 & 0 &1    \\
  \end{array}
\right)$

\end{center}
Meanwhile, we get  that $a_1$, $b_1 \in
G$ \cite{BPS09,NN11}. Let $A_1 = \langle a_1 \rangle$, $ B_1 =
\langle b_1\rangle$ be the cyclic subgroups of $G$ generated by
$a_1$ and $b_1$, respectively. Then, $A_1$ is of order $q^m + 1$ and
$B_1$ is of order $q^{m-1}-1$.

Let $C_1 = \langle b_1^{ \frac{q^{m-1}-1}{q-1}} \rangle$ be the
subgroup of order $q-1$ of $B_1$ and $B_1'=\{gC_1| g\in B_1 \}$ be the left coset of $C_1$ in $B_1$. Thus, $ |B_1'| = \frac{q^m-1}{q-1}$. Furthermore, $A_1$ and $B_1'$ are chosen so that $A_1 \cap B_1'
= \{1\}$ and both are cyclic sets. Then, from Lemma \ref{le:5}, it follows that $A_1$ and $B_1'$ have
MLSs.

Now, let $S_1' = \{W_i'\mid 0 \leq i \leq q^m\}$ be  the classical spread as
described in Remark \ref{remark:Wi}. $W_i'$ are $(m-1)$-dimensional
totally isotropic subspaces of $V$ for $0 \leq i \leq q^m$.  Clearly, the partial spread $S_1'$ partitions
the set of all isotropic points of $P(V)$.

Also, we observe that the group $A_1$ is sharply transitive on $S_1'$ with respect to $W_1'$.
Also, it is clear that $B_1$ is isomorphic to the Singer cyclic subgroup of $GL_{m-1}(q)$ and  $B_1'$ is sharply transitive on $P(W_1'$) with respect to $\langle e_1\rangle$, where $e_1\in W_1'$.

Now, we consider $G^*=SO^{-}_{2m}(q)$. Being similar to the case of $O^{-}_{2m}(q)$, let $a^*_1= x_1^{*q^m-1} \in SO^{-}_{2m}(q)$ be the matrix
corresponding to the linear transformation $T_{\alpha^{q^m-1}}$. Let $W_1'=\{e_1,e_2, \cdots, e_{m-1}\}$ is an $(m-1)$-dimensional totally isotropic subspace of  $V$. Let $D^*_1 \in O_{m-1}(q)\leq GL(W_1')$,
 $b^*_1 \in SO^{-}_{2m}(q)$ is presented well defined as follows:

\begin{center}

$b^*_1$= $\left(
           \begin{array}{cccc}
D^*_1& 0 &0 &0        \\
 0& 1 &0  &0               \\
 0& 0 & D_1^{*t}&0  \\
 0 & 0 & 0 &1              \\
           \end{array}
         \right)$

\end{center}
 thus, $A^*_1 = \langle a^*_1 \rangle$, and $B^*_1 =
\langle b^*_1\rangle$ are the cyclic subgroups of $SO^{-}_{2m}(q)$ generated by
$a^*_1$ and $b^*_1$ with order $q^m + 1$ and $q^{m-1}-1$, respectively.

Let $C^*_1 = \langle b_1^{* \frac{q^{m-1}-1}{q-1}} \rangle$ be the
subgroup of order $q-1$ of $B^*_1$ and $B'^*_1=\{gC^*_1| g\in B^*_1 \}$ be the left coset of $C^*_1$ in $B^*_1$. Thus, $ |B'^*_1| = \frac{q^m-1}{q-1}$. Furthermore, the cyclic sets $A^*_1$ and $B'^{*}_1$ are chosen so that $A^*_1 \cap B'^*_1
= \{1\}$. Consequently, from Lemma \ref{le:5}, it follows that $A^*_1$ and $B'^*_1$ have
MLSs.

Also, let $S_1' = \{W_i'\mid 0 \leq i \leq q^m\}$ be the classical spread as
described in Remark \ref{remark:Wi}. $W_i'$ are $(m-1)$-dimensional
totally isotropic subspaces of $V$ for $0 \leq i \leq q^m$. It's  clear that  the partial spread $S_1'$ partitions
the set of all isotropic points of $P(V)$.

 We observe that the group $A^*_1$ is sharply transitive on $S_1'$ with respect to $W_1'$.
Also,  $B^*_1$ is isomorphic to the \emph{Singer cyclic subgroup} of $GL_{m-1}(q)$ and $B'^*_1$ is sharply transitive on $P(W_1'$) with respect to $\langle e_1\rangle$, where $e_1\in W_1'$.  Hence, we have the following lemma from the fact above.

\begin{lemma} \label{le:7} Let $A_1$, $B_1'\subseteq O_{2m}^{-}(q)$ (resp. $A^*_1$, $B'^*_1\subseteq SO_{2m}^{-}(q)$), $S_1'$ be the partial spread, $W_1'$ be the subspace of $V$ and $w_1 = \langle e_1\rangle$. Then,

\begin{enumerate}
  \item [(i)] $A_1$ (resp. $A^*_1$) is a sharply transitive set on $S_1'$ with respect
to $W_1'$.
  \item [(ii)] $B_1'$ (resp. $B'^*_1$) is a sharply transitive set on $P(W_1')$ with
respect to $w_1$.
\end{enumerate}

\end{lemma}

\begin{theorem}\label{th:1} Let $q$ be a power of odd primes. Then, the
orthogonal group $O_{2m}^{-}(q)$ has an MLS.
\end{theorem}

\begin{proof}
In case, when $m=1$, $G=O_2^{-}(q)\cong D_{q+1}$ which is a dihedral group of order $2(q + 1)$.  From Lemma \ref{le:4},
$O_2^{-} (q)$ has an MLS. When $m>1$, let $A = A_1$, $B = B_1'$, $w_1 = \langle e_1\rangle$, $L$ be the set of all isotropic
points of $P(V)$. Then, from Lemma \ref{le:3} and Lemma \ref{le:7}, $[A_1, B_1', G_w]$ is an LS for $G$. The stabilizer $G_w$ is a semi-direct product of a $p$-group of order $q^{2m-2}$ and $GL_1(q)\times O_{2m-2}^{-}(q)$. Now from Lemma \ref{le:4}, $p$-groups and $GL_1(q)$ have MLSs.
Furthermore, by the induction hypothesis, we assume that $O_{2m-2}^{-}(q)$ has an MLS. Thus, $G_w$ has an MLS. Also, from Lemma \ref{le:5}, $A_1$ and $B_1'$ have MLSs. Hence, using Lemma \ref{le:6}, $G$ has an MLS.
\end{proof}

\begin{theorem}\label{th:2} Let $q$ be a power of odd primes. Then, the special orthogonal group $SO_{2m}^{-}(q)$ has an MLS.

\end{theorem}

\begin{proof}
In case, when $m=1$, $G^*=SO_2^{-}(q)$ is a solvable group of order $q + 1$.  Lemma \ref{le:4} implies that $SO_2^{-} (q)$ has an MLS. When $m>1$, let $A = A^*_1$, $B = B*_1'$, $w_1 = \langle e_1\rangle$, $L$ be the set of all isotropic points of $P(V)$. Thus, from Lemma \ref{le:3} and Lemma \ref{le:7}, $[A^*_1, B'^*_1, G^*_w]$ is an LS for $SO_{2m}^{-}(q)$. The stabilizer $G^*_w$ is a semi-direct product of a $p$-group of order $q^{2m-2}$ and $GL_1(q)\times SO_{2m-2}^{-}(q)$. Now, from Lemma \ref{le:4}, $p$-groups and $GL_1(q)$ have MLSs.
Furthermore, by the induction hypothesis, we assume that $SO_{2m-2}^{-}(q)$ has an MLS. Thus, $G^*_w$ has an MLS. Also, from Lemma \ref{le:5}, $A^*_1$ and $B'^*_1$ have MLSs. Hence, using Lemma \ref{le:6}, $SO_{2m}^{-}(q)$  has an MLS.

\end{proof}

\section{Construction III: MLSs for $O_{2m}^{+}(q)$ and $SO_{2m}^{+}(q)$}

Similarly, we construct MLSs for $O_{2m}^{+}(q)$ and
$SO_{2m}^{+}(q)$. Also,  $L$ is the set of all
isotropic points of $P(V)$.

We at first consider $O_{2m}^{+}(q)$. Let $V = F_{q^{2m}}$, $\alpha$ a primitive element of field $F_{q^{2m}}$ and $q$ be a power of odd primes. 
The corresponding bilinear map and quadratic form are  described as $f(x, y) = f(x_1 + x_2\beta, y_1 + y_2\beta) = tr_ {F_{q^m} /F_q} (x_1 y_2 + x_2 y_1) = \sum _{i=0}^{m-1} (x_1y_2 + x_2 y_1)^{q^i}$ and $Q(x) = Q(x_1 + x_2\beta) = tr_{ F_{q^m} /F_q} (x_1x_2) = \sum _{i=0}^{m-1} (x_1x_2)^{q^i}$,  respectively, where $\beta = \alpha^{q^m-1}$ \cite{NN11}. We observe that the number of non-zero isotropic points in $P(V)$ with respect to the quadratic space $(V, Q)$ are $(q^m - 1)(q^{m-1} + 1)/(q - 1)$ and the quadratic form $Q$ is of plus type \cite{W09,NN11}. Let $G$ be the
group of all isometries of $(V, Q)$, then, $ G \cong O_{2m}^{+}(q)$
and $G$ is a permutation group acting transitively on isotropic
points \cite{W09,NN11}.

 Similarly, we must have to construct a cyclic group $A$ of order $q^{m-1}+1$ which is sharply transitive on a partial spread $S$ and a cyclic set $B$ of cardinality $(q^{m}-1)/(q-1)$ which is sharply transitive on $ P(W)$, the projective subspace of
$P(V)$. Furthermore, we also need to use the Lemma \ref{le:3}, Lemma \ref{le:5} and Lemma \ref{le:6} to prove the
existence of MLSs.

First, we define two special cyclic subgroups of $O_{2m}^{+}(q)$. Let $a_2= x_2^{q^{m-1}-1} \in GL_{2m}(q)$ be the matrix corresponding to the linear transformation $T_{\alpha^{q^{m-1}-1}}$, where
$x_2$=$\left(\begin{array}{ccc}
          x & 0  & 0        \\
        0 & 1  & 0 \\
        0 & 0  & 1    \\
        \end{array}
      \right)$, $x\in GL_{2m-2}(q)$ .
Let $W_2=\{e_1, e_2,\cdots, e_m\}$ be an $m$-dimensional totally isotropic subspace of $V$, $D_2
\in GL(W_2)$ be a generator of the \emph{Singer cyclic subgroup} of
$GL(W_2)$. Then,  $b_2 \in GL_{2m}(q)$ is defined as
follows: \cite{BPS09,NN11}

\begin{center}

$b_2$= $\left(
        \begin{array}{cc}
         D_2  & 0          \\
 0  &  (D_2^t)^{-1}   \\
        \end{array}
      \right)$

\end{center}

Meanwhile, we have that $a_2$, $b_2 \in G$ \cite{BPS09,NN11}.  Let $A_2 = \langle a_2 \rangle$ and $ B_2 =
\langle b_2\rangle$ be the cyclic subgroups of $G$ generated by
$a_2$ and $b_2$, respectively. Then, $A_2$ is of order $q^{m-1} + 1$ and
$B_2$ is of order $q^m-1$.

Let $C_2 = \langle b_2^{ \frac{q^m-1}{q-1}} \rangle$ be the
subgroup of order $q-1$ of $B_2$ and $B_2'=\{gC_2| g\in B_2 \}$ be the left coset of $C_2$ in $B_2$. Thus, $|B_2'| = \frac{q^m-1}{q-1}$. Furthermore, the  cyclic sets$A_2$ and $B_2'$ are chosen so that $A_2 \cap B_2'
= \{1\}$. Then from Lemma \ref{le:5}, it follows that $A_2$ and $B_2'$ have
MLSs.

Now, let $S_2 = \{W_i\mid 0 \leq i \leq q^m\}$ be the classical spread as
described in Remark \ref{remark:Wi}. $W_i$ are m-dimensional
totally isotropic subspaces of $V$ for $0 \leq i \leq q^m$. Then,  the partial spread $S_2$ clearly partitions
the set of all isotropic points of $P(V)$.

Also, we observe that the group $A_2$ is sharply transitive on $S_2$ with respect to $W_2$.
Also, it is clear that $B_2'$ is sharply transitive on $P(W_2)$ with respect to $\langle e_1\rangle$, where $e_1\in W_2$.

 Now, we take account for $G^*=SO_{2m}^{+}(q)$. Being similar to $O_{2m}^{+}(q)$, $a^*_2= x_2^{*q^{m-1}-1} \in SO^{+}_{2m}(q)$ is also the matrix corresponding to the linear transformation $T_{\alpha^{q^{m-1}-1}}$, where
$x^*_2$=$\left( \begin{array}{ccc}
           x^* & 0  & 0 \\
              0 & 1  & 0 \\
              0 & 0  & 1  \\
          \end{array}
        \right)$,  $x^* \in SO^{+}_{2m-2}(q)$.
Let $W_2=\{e_1, e_2,\cdots, e_m\}$ be an $m$-dimensional totally isotropic subspace of $V$ and $D^*_2
\in GL(W_2)$ be a generator of the \emph{Singer cyclic subgroup} of $GL(W_2)$. Then $b^*_2 \in SO^{+}_{2m}(q)$ is defined as
follows: \cite{NN11}

\begin{center}

$b^*_2$=$\left( \begin{array}{cc}
             D^*_2  & 0      \\
 0  &  D_2^{*t}           \\
          \end{array}
        \right)$

\end{center}
Hence, $A_2 = \langle a_2 \rangle$ and $ B_2 = \langle b_2\rangle$ are  the cyclic subgroups of $G$ with order $q^{m-1} + 1$ and $q^m-1$, respectively.

Let $C^*_2 = \langle b_2^{* \frac{q^m-1}{q-1}} \rangle$ be the
subgroup of order $q-1$ of $B^*_2$ and $B^{'*}_2=\{gC^*_2| g\in B^*_2 \}$ be the left coset of $C^*_2$ in $B^*_2$. Thus, $|B'^*_2| = \frac{q^m-1}{q-1}$. Furthermore, the cyclic sets$A^*_2$ and $B'^*_2$ are chosen so that $A^*_2 \cap B'^*_2
= \{1\}$. Then from Lemma \ref{le:5}, it follows that $A^*_2$ and $B'^*_2$ have
MLSs.

Now, let $S_2 = \{W_i\mid 0 \leq i \leq q^m\}$ the classical spread as
described in Remark \ref{remark:Wi}. $W_i$ are m-dimensional
totally isotropic subspaces of $V$ for $0 \leq i \leq q^m$. Thus,, the partial spread $S_2$ clearly partitions
the set of all isotropic points of $P(V)$.

Consequently, we observe that the group $A^*_2$ is sharply transitive on $S_2$ with respect to $W_2$.
Also, it is clear that $B'^*_2$ is sharply transitive on $P(W_2)$ with respect to $\langle e_1\rangle$, where $e_1\in W_2$. Hence, we have
the following lemma.

\begin{lemma} \label{le:8}Let $A_2$, $B_2'\subseteq O_{2m}^{+}(q)$ (resp. $A^*_2$, $B'^*_2\subseteq
SO_{2m}^{+}(q)$ ), $S_2$ be the partial spread, $W_2$ be the subspace of $V$, and $w_2 = \langle e_1\rangle$. Then,

\begin{enumerate}
  \item [(i)] $A_2$ (resp. $A^*_2$) is a sharply transitive set on $S_2$ with respect
to $W_2$.
  \item [(ii)] $B_2'$ (resp. $B'^*_2$) is a sharply transitive set on $P(W_2)$ with
respect to $w_2$.
\end{enumerate}

\end{lemma}

\begin{theorem}\label{th:3} Let $q$ be a power of odd primes. Then, the orthogonal group $O_{2m}^{+}(q)$ has an MLS.

\end{theorem}

\begin{proof}
Let $G =O_{2m}^{+}(q)$. In case,when $m=1$, $O_2^{+}(q)$ is of order $2(q-1)$. Then, by using Lemma \ref{le:4}, $O_2^{+}(q)$ has an MLS. When $m>1$, let $A=A_2$, $B=B_2'$,
$w_2 = \langle e_1\rangle$, $L$ be the set of all isotropic points of $P(V)$.
Hence, from Lemma \ref{le:5} and Lemma \ref{le:8}, $[A_2, B_2',
G_{w_2}]$ is an LS for $G$. Besides, the stabilizer $G_{w_2}$ is a semi-direct product of a $p$-group of
order $q^{2m-2}$ and $GL_1(q)\times O_{2m-2}^{+}(q)$. Now from Lemma \ref{le:4}, $p$-groups and $GL_1(q)$
have MLSs. Furthermore, by the induction hypothesis, we assume
that $O_{2m-2}^{+}(q)$ has an MLS,  therefore, $O_{2m}^{+}(q)$ also has an
MLS. Thus, $G_{w_2}$ has an MLS. Also, from Lemma \ref{le:5}, the cyclic sets $A_2$ and $B_2'$ have
MLSs. Therefore, using Lemma \ref{le:6}, $G$ has an MLS.
\end{proof}

\begin{theorem} \label{th:4} Let $q$ be a power of odd primes. Then, special orthogonal group $SO_{2m}^{+}(q)$ has an MLS.

\end{theorem}

\begin{proof}
 In case, when $n=1$,  $G^*=SO_2^{+}(q)$ is a solvable group of order $q-1$. Then by using Lemma \ref{le:4}, $SO_2^{+}(q)$ has an MLS. When $m>1$, let $A=A^*_2$, $B=B'^*_2$,
$w_2 = \langle e_1\rangle$, $L$ be the set of all isotropic points of $P(V)$. Hence, from Lemma \ref{le:5} and Lemma \ref{le:8}, $[A^*_2, B'^*_2, G^*_{w_2}]$ is an LS for $G$. Besides, the stabilizer $G^*_{w_2}$ is a semi-direct product of a $p$-group of order $q^{2m-2}$ and $GL_1(q)\times SO_{2m-2}^{+}(q)$. Now from Lemma \ref{le:4}, $p$-groups and $GL_1(q)$
have MLSs. Furthermore, by the induction hypothesis, we assume that $SO_{2m-2}^{+}(q)$ has an MLS, therefore,  $G^*_{w_2}$ has an MLS. Also, from Lemma \ref{le:5}, the cyclic sets $A^*_2$ and $B'^*_2$ have MLSs.  Finally,, using Lemma \ref{le:6},  $SO_{2m}^{+}(q)$ has an MLS.

\end{proof}

\section{Construction IV: MLSs for $PSO_{2m}^{\pm}(q)$ and $P\Omega_{2m}^{\pm}(q)$}
In this section, we consider the MLSs for $PSO_{2m}^{\pm}(q)$ and $P\Omega_{2m}^{\pm}(q)$. Being different from $O^{\pm}_{2m}(q)$ and $SO^{\pm}_{2m}(q)$, our technique is based on some canonical homomorphisms.

\begin{theorem} \label{th:5} Let $q$ be a power of odd primes. Then, $PSO^{\pm}_{2m}(q)$ has an MLS.

\end{theorem}

\begin{proof}

(1)In case, when $G^*=SO^{-}_{2m}(q)$ and $G'= PSO^{-}_{2m}(q)$, let $A = A^*_1$, $B = B'^*_1$, $w_1 = \langle e_1\rangle$, $L$ be the set of all isotropic points of $P(V)$ as described as Section 3. Suppose $\eta_1:SO^{-}_{2m}(q) \rightarrow PSO^{-}_{2m}(q)\cong SO^{-}_{2m}(q)/Z(SO^{-}_{2m}(q)$ is the canonical homomorphism onto $PSO^{-}_{2m}(q)$, and let $\overline{A^*_1}=\eta(A^*_1)$, $\overline{B'^*_1}=\eta(B'^*_1)$ and $\overline{G^*_{w_1}}=\eta(G^*_{w_1})$, then [$\overline{A^*_1}$, $\overline{B'^*_1}$, $\overline{G^*_{w_1}}$] is the corresponding LS for $PSO^{-}_{2m}(q)$ from Lemma \ref{le:2}.
Also from Section 3, the stabilizer $\overline{G^*_{w_1}}$ is a semi-direct product of a $p$-group of order $q^{2m-2}$ and
$GL_1(q)\times PSO_{2m-2}^{-}(q)$. Thus, using the same induction as used in Theorem \ref{th:2}, we get that $PSO^{-}_{2m}(q)$ has an MLS.

(2)In case, when $G^*=SO^{+}_{2m}(q)$ and $G'= PSO^{+}_{2m}(q)$, let $A = A^*_2$, $B = B'^*_2$, $w_2 = \langle e_1\rangle$, $L$ be the set of all isotropic points of $P(V)$ as described as Section 4. Suppose $\eta_2:SO^{+}_{2m}(q) \rightarrow PSO^{+}_{2m}(q)\cong SO^{+}_{2m}(q)/Z(SO^{+}_{2m}(q)$ is the canonical homomorphism onto $PSO^{+}_{2m}(q)$, and let $\overline{A^*_2}=\eta(A^*_2)$, $\overline{B'^*_2}=\eta(B'^*_2)$ and $\overline{G^*_{w_2}}=\eta(G^*_{w_2})$, then [$\overline{A^*_2}$, $\overline{B'^*_2}$, $\overline{G^*_{w_2}}$] is the corresponding LS for $PSO^{+}_{2m}(q)$ from Lemma \ref{le:2}.
Hence, the stabilizer $\overline{G^*_{w_2}}$ is a semi-direct product of a $p$-group of order $q^{2m-2}$ and
$GL_1(q)\times PSO_{2m-2}^{+}(q)$. Thus, using the same induction as used in Theorem \ref{th:2}, we get that $PSO^{+}_{2m}(q)$ has an MLS.

\end{proof}

\begin{theorem} \label{th:6} Let $q$ be a power of odd primes. Then, $P\Omega^{\pm}_{2m}(q)$  has an MLS.

\end{theorem}

\begin{proof}

(1)In case, when $q^m\equiv -1 \bmod 4$, from \cite{W09},we have that $PSO^{\pm}_ {2m}(q)=P\Omega^{\pm}_{2m}(q)$.  Therefore,  being similar to the case in $PSO^{\pm}_{2m}(q)$, [$\overline{A^*_1}$, $\overline{B'^*_1}$, $\overline{G^*_{w_1}}$] is the corresponding LS for $P\Omega^{-}_{2m}(q)$ and [$\overline{A^*_2}$, $\overline{B'^*_2}$, $\overline{G^*_{w_2}}$] is the corresponding LS for $P\Omega^{+}_{2m}(q)$.  Meanwhile, as described in Theorm \ref{th:5},  $P\Omega^{-}_{2m}(q)$ and $P\Omega^{+}_{2m}(q)$ both have MLSs.

(2)In case, when $q^m\equiv 1 \bmod 4$,  we must consider the reflection in $G^*=SO_n(q)$.  For  the  isotropic 1-space $ w=\langle e_1\rangle$,  the reflection $r_w: V \rightarrow V$  is defined by $r_{w}(v)=v-2 \frac{f(v,w)}{f(w,w)}w $ for each $v\in V$. Also,  the linear transformation $r_w$  is an element of $G^*_w$.  For$G'=\Omega^{-}_{2m}(q)$  (resp. $\Omega^{+}_{2m}(q)$),  each element of $G'_w$  is a product of an even number of reflections and   $G'_w$  is a semi-direct product of a $p$-group of order $q^{2m-2}$ and $GL_1(q)\times \Omega^{-}_{2m-2}(q)$ (resp. $GL_1(q)\times \Omega^{+}_{2m-2}(q)$). Besides,  an element $x$ in $SO^{\pm}_ {2m}(q)$ is in $\Omega^{\pm}_{2m}(q)$  if and only if the rank of $I_{2m} + x$ is even.  From the construction of  $A^*_1$ (resp. $A^*_2 $)  and  $B'^*_1$ (resp. $B'^*_2 $) in Section 3 and Section 4, we get that  the ranks of $I_{2m} + x^*_1$ (resp. $I_{2m} + x^*_2$) and  $I_{2m} + b^*_1$ (resp. $I_{2m} + b^*_2$ )  are both even, Thus,  $A^*_1$ (resp. $A^*_2 ) \leq G'$, $B'^*_1$ (resp. $B'^*_2 ) \subseteq G'$. As described in Theorem \ref{th:3} and Theorem \ref{th:4}, we have  $\Omega^{-}_{2m}(q)$ (resp. $\Omega^{+}_{2m}(q)$) has an MLS.

Furthermore,, let $\theta_1:  \Omega^{-}_ {2m}(q) \rightarrow P\Omega^{-}_{2m}(q)$ and  $\theta_2:  \Omega^{+}_ {2m}(q) \rightarrow P\Omega^{+}_{2m}(q)$ be  the canonical homomorphisms onto $P\Omega^{-}_{2m}(q)$ and $P\Omega^{+}_{2m}(q)$, respectively.  Therefore,   [$\theta_1(A^*_1)$, $\theta_1(B'^*_1)$, $\theta_1(G^*_{w_1})$] is the corresponding LS for $P\Omega^{-}_{2m}(q)$ and  [$\theta_2(A^*_2)$, $\theta_2(B'^*_2)$, $\theta_2(G^*_{w_2})$] is the corresponding LS for $P\Omega^{+}_{2m}(q)$ . Also,  for   $P\Omega^{-}_{2m}(q)$, the stabilizer $\theta_1(G^*_{w_1})$  is a semi-direct product of a $p$-group of order $q^{2m-2}$ and $GL_1(q)\times P\Omega^{-}_{2m-2}(q) $;  for  $P\Omega^{+}_{2m}(q)$. the stabilizer  $\theta_2(G^*_{w_2})$  is a semi-direct product of a $p$-group of order $q^{2m-2}$ and $GL_1(q)\times P\Omega^{+}_{2m-2}(q) $. Thus, using the same induction as used in Theorem \ref{th:2} and Theorem \ref{th:4}, we get that $P\Omega^{-}_{2m}(q)$ and $P\Omega^{+}_{2m}(q)$ both have MLSs.

\end{proof}

\section{Construction V: MLSs for $O_{2m+1}(q)$, $SO_{2m+1}(q)$, $PSO_{2m+1}(q)$ and $P\Omega_{2m+1}(q)$}

In this section, we first construct MLSs for $O_{2m+1}(q)$ and $SO_{2m+1}(q)$. Then we consider the MLSs for $PSO_{2m+1}(q)$ and $P\Omega_{2m+1}(q)$. Here, $L$ is also the set of all isotropic points of $P(V)$.

Let $V = F_{q^{2m+1}}$ and $q$ be a power of odd primes. 
The corresponding non-singular alternating bilinear map is $f(x, y) = tr_{F_ {q^{2m+1}}/F_q} (ax\overline{y}) = \sum _{i=1}^{2m+1} (ax \overline{y})^{q^i}$, where $a \in F_{q^{2m+1}}^*$ and $a + \overline{a} = 0$ \cite{NN11}. Also,
$G = O_{2m+1}(q)$ is the isometry group of the inner product space $(V, f )$ and $G$ is a permutation group acting transitively on isotropic points \cite{W09,NN11}. Then we observe that the number of non-zero isotropic points in $P(V)$ with respect to the inner product space $(V, f )$ are $(q^m - 1)(q^m + 1)/(q - 1)$ \cite{W09,NN11}.


 Similarly, we construct a cyclic group $A$ of order $q^m+1$, which is sharply transitive on a partial spread $S$ and a cyclic set $B$ of cardinality $q^{m}-1/(q-1)$, which
is sharply transitive on $ P(W)$, the projective subspace of
$P(V)$. Then we use the Lemma \ref{le:3}, Lemma \ref{le:5} and Lemma \ref{le:6} to prove the
existence of MLSs.

Let $a_3= x_3^{q^m-1} \in GL_{2m+1}(q)$ be the matrix corresponding to the linear transformation $T_{\alpha^{q^m-1}}$, where
$x_3$= $\left(
         \begin{array}{cc}
           x & 0      \\
           0 & 1      \\
         \end{array}
       \right)$,  $x\in GL_{2m}(q)$. Let $W_3=\{e_1, e_2,\cdots, e_m\}$ be an $m$-dimensional totally isotropic subspace of $V$, $D_3 \in GL(W_3)$ be a generator of the \emph{Singer cyclic subgroup} of
$GL(W_3)$. Then $b_3 \in GL_{2m+1}(q)$ defined as
follows: \cite{NN11}

\begin{center}

$b_3$=$\left(\begin{array}{ccc}
         D_3 &  0 & 0         \\
     0   &  (D_3^t)^{-1} &  0 \\
       0   &  0  & 1      \\
        \end{array}
      \right)$

\end{center}

From \cite{BPS09,NN11}, we have that $a_3$, $b_3 \in G$. Let $A_3 = \langle a_3 \rangle$, $ B_3 =
\langle b_3\rangle$ be cyclic subgroups of $G$. Thus, $|A_3|=q^m + 1$ and $|B_3|=q^m-1$.

Let $C_3 = \langle b_3^{ \frac{q^m-1}{q-1}} \rangle$ be the
subgroup of order $q-1$ of $B_3$, $B_3'=\{gC_3| g\in B_3 \}$ the left coset of $C_3$ in $B_3$. Therefore, $|B_3'| = \frac{q^m-1}{q-1}$. Furthermore, the cyclic sets $A_3$ and $B_3'$ are chosen so that $A_3 \cap B_3'= \{1\}$. Hence,  from Lemma \ref{le:5}, it follows that $A_3$ and $B_3'$ have
MLSs.

Now, let $S_3 = \{W_i\mid 0 \leq i \leq q^m\}$ be the classical spread as
described in Remark \ref{remark:Wi}, $W_i$ be $m$-dimensional
totally isotropic subspaces of $V$ for $0 \leq i \leq q^m$. Therefore, the partial spread $S_3$ partitions
the set of all isotropic points of $P(V)$.

We observe that the group $A_3$ is sharply transitive on $S_3$ with respect to $W_3$.
Also, $B_3'$ is sharply transitive on $P(W_3)$ with respect to $\langle e_1\rangle$, where $e_1\in W_3$.

In case, when $G^*=SO_{2m+1}(q)$, let $a^*_3= x_3^{* q^m-1} \in SO_{2m+1}(q)$ be the matrix corresponding to the linear transformation $T_{\alpha^{q^m-1}}$, where
$x^*_3$=$\left(
          \begin{array}{cc}
            x^* & 0      \\
                0 & 1      \\
          \end{array}
        \right)$,  $x^* \in SO^{\pm}_{2m}(q)$. Let $W_3=\{e_1, e_2,\cdots, e_m\}$ be an $m$-dimensional totally isotropic subspace of $V$ and $D^*_3 \in O_m(q) \leq GL(W_3)$. Then, $b^*_3 \in SO_{2m+1}(q)$ is defined as
follows: \cite{NN11}

\begin{center}

$b_3$=$\left( \begin{array}{ccc}
         D^*_3 &  0 & 0         \\
 0   &  D_3^{* t}&  0 \\
 0   &  0  & 1  \\
        \end{array}
      \right)$

\end{center}

From \cite{BPS09,NN11}, we have that $a^*_3$, $b^*_3 \in G$. Let $A^*_3 = \langle a^*_3 \rangle$, $ B^*_3 =
\langle b^*_3\rangle$ be cyclic subgroups of $G$. Then, $|A^*_3|=q^m + 1$ and $|B^*_3|=q^m-1$.

Let $C^*_3 = \langle b_3^{*\frac{q^m-1}{q-1}} \rangle$ be the
subgroup of order $q-1$ of $B^*_3$ and $B^{'*}_3=\{gC^*_3| g\in B^*_3 \}$ be the left coset of $C^*_3$ in $B^*_3$. Thus, $|B'^*_3| = \frac{q^m-1}{q-1}$. Furthermore, the cyclic sets $A^*_3$ and $B'^*_3$ are chosen so that $A^*_3 \cap B'^*_3= \{1\}$. Then from Lemma \ref{le:5}, it follows that $A^*_3$ and $B'^*_3$ have MLSs.

Now, let $S_3 = \{W_i\mid 0 \leq i \leq q^m\}$ be the classical spread as
described in Remark \ref{remark:Wi}, $W_i$ be $m$-dimensional
totally isotropic subspaces of $V$ for $0 \leq i \leq q^m$. Then, the partial spread $S_3$ partitions
the set of all isotropic points of $P(V)$.

Moreover, we observe that the group $A^*_3$ is sharply transitive on $S_3$ with respect to $W_3$.
Also, $B'^*_3$ is sharply transitive on $P(W_3)$ with respect to $\langle e_1\rangle$, where $e_1\in W_3$. Hence, we have the following lemma.

\begin{lemma} \label{le:9}
Let $A_3$, $B_3'\subseteq O_{2m+1}(q)$ (resp. $A^*_3$, $B'^*_3\subseteq SO_{2m+1}(q)$), $S_3$ be the partial spread, $W_3$ be the subspace of $V$, and $w = \langle e_1\rangle$. Then,

\begin{enumerate}
  \item [(i)] $A_3$ (resp. $A^*_3$) is a sharply transitive set on $S_3$ with respect
to $W_3$.
  \item [(ii)] $B_3'$ (resp. $B'^*_3$) is a sharply transitive set on $P(W_3)$ with
respect to $w_3$.
\end{enumerate}

\end{lemma}

\begin{theorem}\label{th:7} Let $q$ be a power of odd primes. Then, the orthogonal group $O_{2m+1}(q)$ has an MLS.

\end{theorem}

\begin{proof}
Let $G =O_{2m+1}(q)$, $A=A_3$, $B=B_3'$, $w = \langle e_1\rangle$. When $m=0$, $O_1(q)\cong C_2$. Lemma \ref{le:4} implies that $O_1(q)$ has an MLS. When $m\geq 1$, from Lemma \ref{le:5} and Lemma \ref{le:9}, $[A_3, B_3',
G_w]$ is an LS for $G$. Hence, the stabilizer $G_w$ is a semi-direct product of a $p$-group of
order $q^{2m-1}$ and $GL_1(q)\times O_{2m-1}(q)$. Now from Lemma \ref{le:4}, $p$-groups and $GL_1(q)$
have MLSs. Furthermore, by the induction hypothesis, we assume
that $O_{2m-1}(q)$ has an MLS, therefore, $O_{2m+1}(q)$ also has an
MLS. Thus, $G_w$ has an MLS. Also, from Lemma \ref{le:5}, the cyclic sets $A_3$ and $B_3'$ have
MLSs. Therefore, using Lemma \ref{le:6}, $G$ has an MLS.
\end{proof}

\begin{theorem} \label{th:8} Let $q$ be a power of odd primes. Then, $SO_{2m+1}(q)$  has an MLS.

\end{theorem}

\begin{proof}
Let $G^* =SO_{2m+1}(q)$, $A=A^*_3$, $B=B'^*_3$, $w = \langle e_1\rangle$. In case, when $m=0$, $SO_1(q)\cong C_2$.  Lemma \ref{le:4} implies that $SO_1(q)$ has an MLS. In case, when $m\geq 1$, from Lemma \ref{le:5} and Lemma \ref{le:9}, $[A^*_3, B'^*_3, G_w]$ is an LS for $G$. Then the stabilizer $G^*_w$ is a semi-direct product of a $p$-group of
order $q^{2m-1}$ and $GL_1(q)\times SO_{2m-1}(q)$. Now from Lemma \ref{le:4}, $p$-groups and $GL_1(q)$ have MLSs. Furthermore, by the induction hypothesis, we assume
that $SO_{2m-1}(q)$ has an MLS, therefore, $G^*_w$ has an MLS. Also, from Lemma \ref{le:5}, the cyclic sets $A^*_3$ and $B'^*_3$ have MLSs. Therefore, using Lemma \ref{le:6},  $SO_{2m+1}(q)$ has an MLS.

\end{proof}

\begin{theorem} \label{th:6.3} Let $q$ be a power of odd primes. Then, $PSO_{2m+1}(q)$ has an MLS.

\end{theorem}

\begin{proof}
In case, when $G^*=SO_{2m+1}(q)$ and $G'^*= PSO_{2m+1}(q)$, let $A = A^*_3$, $B = B'^*_3$, $w = \langle e_1\rangle$, $L$ be the set of all isotropic points of $P(V)$ as described above. Suppose $\eta_3:SO_{2m+1}(q) \rightarrow PSO_{2m+1}(q)\cong SO_{2m+1}(q)/Z(SO_{2m+1}(q))$ is the canonical homomorphism onto $PSO_{2m+1}(q)$, and let $\overline{A}^*_3=\eta(A^*_3)$, $\overline{B}^*_3=\eta(B^*_3)$ and $\overline{G_w}=\eta(G_w)$, then [$\overline{A}^*_3$, $\overline{B}^*_3$, $\overline{G_w}$] is the corresponding LS for $PSO_{2m+1}(q)$ from Lemma \ref{le:2}. Also from Section 3, the stabilizer $\overline{G_w}$ is a semi-direct product of a $p$-group of order $q^{2m-1}$ and $GL_1(q)\times  PSO_{2m-1}(q)$.
Thus, using the same induction as used in Theorem \ref{th:2}, we get that $PSO_{2m+1}(q)$ has an MLS.

\end{proof}

\begin{theorem} \label{th:6.3} Let $q$ be a power of odd primes. Then, $P\Omega_{2m+1}(q)$ has an MLS.

\end{theorem}

\begin{proof}
In case, when $n=2m+1$,  from \cite{W09}, we observe that  $\Omega_{2m-1}(q) \cong Sp_{2m-2}(q) $.  Also from \cite{NN11}, we have that  $Sp_{2m-2}(q)$ has an MLS.  Using the induction, we get that  $\Omega_{2m+1}(q) $ has an MLS. Then, we utilize the canonical homomorphism $\theta_3: \Omega_{2m+1} (q) \rightarrow P\Omega_{2m+1}(q)$ for proving that $P\Omega_{2m+1}(q)$ has LS . Consequently, using the same induction as used in Section 5, we get that $P\Omega_{2m+1}(q)$ has an MLS. Here, we omit the corresponding  proof.
\end{proof}

\section*{Conclusion}
We utilize partial spreads of totally isotropic subspaces, stabilizers of isotopic 1-subspaces and linear transformations in corresponding vector spaces to construct MLSs for $O_n(q)$, $SO_n(q)$, $PO_n(q)$, $PSO_n(q)$, $\Omega_n(q)$ and $P\Omega_n(q)$ with $q$ as a power of odd primes.  Meanwhile, our methods can be used to construct MLSs for other finite simple groups.

\section*{Acknowledgements}
This work is partially supported by the National Natural Science Foundation of China (NSFC) (Nos.61103198, 61121061,61370194) and the NSFC A3 Foresight Program (No.61161140320)

\end{document}